%% file: paper.tex
\documentclass[11pt]{article}
\usepackage{rheaj}
\usepackage{setspace}

\newcommand{\racke}{R\"{a}cke\xspace}
\newcommand{\raecke}{R\"{a}cke\xspace}

\newcommand{\load}{\textnormal{load}}
\newcommand{\rload}{\textnormal{rload}}
\newcommand{\maxflow}{\textnormal{maxflow}}

\title{Approximating Flexible Graph Connectivity  \\
via \racke Tree based Rounding}

\author{
     Chandra Chekuri\thanks{Dept.\ of Computer Science, Univ.\ of Illinois,
    Urbana-Champaign, Urbana, IL 61801. \texttt{chekuri@illinois.edu}. Supported in part by NSF grants
     CCF-1910149 and CCF-1907937.}
   \and
  Rhea Jain\thanks{Dept.\ of Computer Science, Univ.\ of Illinois,
    Urbana-Champaign, Urbana, IL 61801. \texttt{rheaj3@illinois.edu}. Supported in part by NSF grant
    CCF-1910149.}}
    
\date{\today}

\begin{document}

\maketitle

\input{abstract}

\input{introduction}

\input{preliminaries}

\input{algorithm}

\input{analysis}

\paragraph{Acknowledgements} We thank Qingyun Chen for
clarifications on a proof in \cite{ChenLLZ22}.

\bibliographystyle{plainurl}
\bibliography{fgc}

\end{document}

%% file: abstract.tex
\begin{abstract}
Flexible graph connectivity is a new network design model introduced
by Adjiashvili \cite{Adjiashvili13}. It has seen several recent
algorithmic advances
\cite{AdjiashviliHM20,AdjiashviliHM22,AdjiashviliHMS20,BoydCHI22,BansalCGI22,ChekuriJ22}. Despite
these, the approximability even in the 
setting of a single-pair $(s,t)$ is poorly understood. In
\cite{ChekuriJ22} we raised the question of whether there is
poly-logarithmic approximation for the survivable network design
version (Flex-SNDP) when the connectivity requirements are fixed
constants. In this paper, we adapt a powerful framework for
survivable network design recently developed by Chen, Laekhanukit,
Liao, and Zhang \cite{ChenLLZ22} to give an affirmative answer to the
question. The framework of \cite{ChenLLZ22} is based on \raecke trees
and group Steiner tree rounding. The algorithm and analysis also
establishes an upper bound on the integrality gap of an
LP relaxation for Flex-SNDP \cite{ChekuriJ22}.
\end{abstract}

%% file: introduction.tex
\section{Introduction}
The Survivable Network Design Problem (SNDP) is an important problem in
combinatorial optimization that generalizes many well-known problems
related to connectivity, and is also motivated by practical problems
related to the design of fault-tolerant networks. The input to this
problem is an undirected graph $G=(V,E)$ with non-negative edge costs $c: E \to \R_+$ and an integer function $r: V
\times V \to \Z^+$ which specifies a connectivity requirement for each
node pair $(u,v)$. The goal is to find a minimum-cost subgraph $H$ of
$G$ such that $H$ has $r(u,v)$ connectivity for each pair $(u,v)$. Our
focus in this paper is on edge-connectivity requirements; the resulting
problem is referred to as EC-SNDP. VC-SNDP refers to the problem in
which each pair $(u,v)$ requires $r(u,v)$ vertex connectivity. 
EC-SNDP contains as special cases classical problems such
as $s$-$t$ shortest path, minimum spanning tree (MST), minimum
$k$-edge-connected subgraph ($k$-ECSS), Steiner tree, Steiner forest
and several others. It is NP-Hard and APX-Hard to approximate.
Jain's seminal $2$-approximation for EC-SNDP via iterated rounding \cite{Jain01}
is the currently the best known approximation ratio.

In this paper we are interested in a new network design model
suggested by Adjiashvili \cite{Adjiashvili13} for which there are
several recent developments
\cite{AdjiashviliHM20,AdjiashviliHM22,AdjiashviliHMS20,BoydCHI22,ChekuriJ22,BansalCGI22}.
In this model, the edge set $E$ is partitioned to \emph{safe} edges
$\calS$ and \emph{unsafe} edges $\calU$. Vertices $s, t \in V$ are
$(p,q)$-flex-connected
\footnote{We follow the terminology from our
recent work \cite{ChekuriJ22} that is influenced by
\cite{AdjiashviliHM22,BoydCHI22}.}
if $s$ and $t$ are
$p$-edge-connected after deleting any subset of at most $q$ unsafe
edges. The Flex-SNDP problem is the following: the input is a graph
$G=(V,E)$ with edge costs $c:E \rightarrow \mathbb{R}_+$, a partition
$\calU \uplus \calS$ of the edge set, and functions $p, q: V \times V
\to \Z^+$. The goal is to find a min-cost subgraph $H$ of $G$ such
that each $u, v \in V$ is $(p(u, v), q(u, v))$-flex-connected in
$H$. We denote by $(p,q)$-Flex-SNDP the special case where for each
vertex pair $u,v$, either $p(u,v) = q(u,v) = 0$ or $p(u,v) = p$,
$q(u,v) = q$. Note that if all edges are safe, i.e. $E = \calS$, then
$(p, q)$-flex-connectivity is the same as $p$-connectivity, and if all
edges are unsafe, i.e. $E = \calU$, then $(p, q)$-flex-connectivity is
the same as $(p+q)$-connectivity. Flex-SNDP thus generalizes EC-SNDP.

The work so far in flexible connectivity has been on two special
cases. The first is the spanning case, which requires
$(p,q)$-flex-connectivity for all pairs of vertices. This is the
$(p,q)$-FGC problem \cite{BoydCHI22}. The other is when
the requirement is for a single pair $(s,t)$
\cite{Adjiashvili13,AdjiashviliHMS20}.  Following \cite{ChekuriJ22} we
refer to this as $(p,q)$-Flex-ST. For $(p,q)$-FGC,  \cite{BoydCHI22}
obtained an $O(q \log n)$-approximation, and constant factor
approximations have been developed for small values of $p,q$
\cite{BoydCHI22,BansalCGI22,ChekuriJ22}. For $(p,q)$-Flex-ST,
the only non-trivial approximations known are for $(1,q)$-Flex-ST and
$(p,1)$-Flex-ST \cite{AdjiashviliHMS20} and $(2,2)$-Flex-ST
\cite{ChekuriJ22}. In fact, no non-trivial approximation is known even for
$(3,2)$-Flex-ST or $(2,3)$-Flex-ST. No non-trivial result is known for
$(2,2)$-Flex-SNDP. We refer the reader to \cite{BoydCHI22,ChekuriJ22,BansalCGI22} for
a more detailed description of existing results.

Adjiashvili et al.\ \cite{AdjiashviliHMS20} show that when $p$ is part of the input and
large, $(p,1)$-Flex-ST is NP-Hard to approximate to almost polynomial
factors. Thus $(p,q)$-Flex-SNDP is a harder problem than EC-SNDP. In
our earlier paper \cite{ChekuriJ22} we raised the following
question: does $(p,q)$-Flex-SNDP admit an approximation ratio of the
form $f(p,q)$ or $f(p,q)\text{polylog}(n)$ for for all fixed $p,q$
where $f$ is some integer valued function? \cite{ChekuriJ22}
formulated an LP relaxation which can be solved in $n^{O(q)}$-time and
a corresponding question on its integrality gap was also implicitly
raised. The known techniques for EC-SNDP and related problems rely
on the requirement function satisfying structural
properties such as skew-supermodularity and uncrossability. These properties are crucial in
primal-dual and iterated rounding based algorithms \cite{GoemansGPSTW94,Jain01}. Recent
work on flexible connectivity network design
\cite{BoydCHI22,BansalCGI22,ChekuriJ22} extended some of these ideas
in non-trivial and interesting ways to the special cases that we
mentioned. However, the requirement function for flexible connectivity
is not as well-behaved (see \cite{ChekuriJ22} for some examples) and
it seems challenging to obtain any non-trivial algorithm for say
$(2,2)$-Flex-SNDP or $(3,2)$-Flex-ST.

\paragraph{Contribution:} In this paper we take a substantially different approach for
$(p,q)$-Flex-SNDP, motivated by a very recent work of Chen,
Laekhanukit, Liao, and Zhang \cite{ChenLLZ22}. They developed a new
algorithmic approach for survivable network design to tackle a
generalization of EC-SNDP to the group/set connectivity setting.
We use their framework to obtain the following theorem.

\begin{theorem}
\label{thm:main}
There is a randomized algorithm that yields an $O(q(p+q)^3 \log^7
n)$-approximation for $(p,q)$-Flex-SNDP and runs in expected
$n^{O(q)}$-time. The approximation is based on an LP relaxation for the problem.
\end{theorem}

\begin{remark}
  The algorithm for $(p,q)$-Flex-SNDP easily extends to the setting where
  the maximum connectivity requirement is dominated by $(p,q)$. 
\end{remark}

The preceding theorem sheds light on the approximability of the
problem --- as discussed, this has been challenging via past
techniques.  It suggests that there may be an $f(p,q)$-approximation
for $(p,q)$-Flex-SNDP via the LP relaxation. It also showcases
the generality of the approach in \cite{ChenLLZ22} which
is likely to have further impact in network design. For instance,
the algorithm and analysis extend to the Set Connectivity version of
flexible connectivity problem.

\subsection{Technical Overview}
We give a brief technical overview of the algorithm. We follow an
augmentation approach for $(p,q)$-Flex-SNDP following recent work
\cite{BoydCHI22,BansalCGI22,ChekuriJ22}. The idea is to start with a
subgraph $H_0$ that satisfies $(p,0)$-flex-connectivity for the given
instance (which can be solved via reduction to EC-SNDP) and
iteratively increase, in $q$ stages, to obtain a subgraph $H_q$ that
satisfies $(p,q)$-flex-connectivity. In stage $\ell$ the goal is to go
from $(p,\ell)$-flex-connectivity to
$(p,\ell+1)$-flex-connectivity. Call a set $S \subset V$ deficient in
stage $\ell$ if it separates some terminal pair and has the following
property: $|\delta_{H_{\ell}}(S)| = p+\ell$ and $|\delta_{H_{\ell}}(S)
\cap \calS| < p$ (has less than $p$ safe edges).  It is necessary and
sufficient to cover all the deficient cuts by any edge in
$E(G)-E(H_\ell)$ to increase to $(p,\ell+1)$-flex-connectivity. Thus,
the augmentation problem can ignore the distinction between safe and
unsafe edges.  The family of deficient cuts in this augmentation
problem, unfortunately, does not have nice properties such as
uncrossability except in some special cases. Instead, we rely on the
framework of \cite{ChenLLZ22} --- we start with a fractional solution
to a natural cut covering relaxation that can be solved in $n^{O(q)}$
time, and round it via their approach. The algorithm in
\cite{ChenLLZ22} consists of three main ingredients. The first is to
use the fractional solution $x$ to define a capacitated graph $G'$ in
a clever way. The second is to consider a probabilistic approximation
of $G'$ via capacitated trees that approximate the flow properties of
$G'$ as defined in the seminal work of \racke \cite{Racke08}. The
third is a dependent randomized rounding procedure on trees for the
group Steiner tree problem due to Garg, Konjevod and Ravi
\cite{GargKR98}; this rounding has been generalized to the Set
Connectivity problem that we will need \cite{ChekuriEGS11,CGL15}. The
overall algorithm of \cite{ChenLLZ22} is simple at a high-level. After
it sets up the graph $G'$, it repeatedly samples a tree from the
\racke tree distribution induced by $G'$, and does a randomized
oblivious Set Connectivity rounding on the sampled tree. We use the
same algorithm with relevant changes to the definition of capacities
of $G'$ that are tailored to flexible connectivity. The analysis in
\cite{ChenLLZ22} shows that the rounding procedure covers the
deficient cuts in only a polylogarithmic number of rounds, via a
clever argument. We adapt their analysis to show that it also works
for flexible connectivity.  An important comment is the following: the
algorithm in \cite{ChenLLZ22} was designed to address group/set
connectivity where the complexity comes from large groups (otherwise
it can be reduced to EC-SNDP). Even though we are solving a non-group
problem in $(p,q)$-Flex-SNDP, the algorithm and analysis shows that
ideas from set connectivity implicitly arise when one tries to connect
different components in the process of covering deficient cuts.

\paragraph{Organization:} Section~\ref{sec:prelim} sets up the relevant background on
the LP relaxation, \racke tree embeddings, group Steiner tree, Set Connectivity and
the tree rounding algorithm for them. Section~\ref{sec:algo} describes the rounding algorithm.
Section~\ref{sec:analysis} analyses the correctness and approximation ratio of the algorithm.

%% file: preliminaries.tex
\section{Preliminaries, Augmentation LP, and Background}
\label{sec:prelim}

Let $G=(V,E)$ be a graph with edge capacities $x: E
\rightarrow \mathbb{R}_+$, and let $S,T$ be two disjoint vertex
subsets. We say that $x$ supports a flow of value $f$ between $A$ and
$B$ if, in the graph $G'$ obtained by shrinking $S$ to $s$ and $T$ to
$t$, the max $s$-$t$ flow is at least $f$.

\subsection{LP Relaxation for Augmentation}
Recall the definition of the $(p,q)$-Flex-SNDP problem: given a graph
$G = (V, E)$ with edge costs $c(e)$, a partition $\calU \uplus \calS$
on the edge set, and terminal pairs $(s_i, t_i) \in V \times V$, $i
\in [k]$, find the cheapest subgraph such that each terminal pair
$(s_i, t_i)$ is $(p,q)$-flex-connected. Equivalently, any cut $\delta(S)$
separating a terminal pair must have at least $p$ safe edges or at least $p +
q$ total edges. One can verify the feasibility of a solution in 
$n^{O(q)}$-time: for each subset of $q$ unsafe edges, remove them
and test for $p$-connectivity between terminal pairs.
We employ the augmentation methodology. Suppose we are given a partial solution $H
\subseteq E$ that satisfies $(p,q-1)$-flex-connectivity for the given
terminal pairs. We call a cut $S$ \emph{violated} with respect to $H$
if $|S \cap \{s_i,t_i\}| = 1$ for some $i$, and $|\delta_H(S) \cap \calU| < p$
  and $|\delta_H(S)| = p+q-1$. Note that any cut separating a
terminal pair in $H$ has at least $p$ safe edges or at least $p+q-1$ total
edges. Let $\calC = \{S \subseteq V:
S \text{ is violated}\}$.  We
can augment $H$ to obtain a feasible solution to $(p,q)$-Flex-SNDP
instance by covering all cuts in $\calC$.  This naturally leads to a cut-based
LP relaxation for the augmentation problem with variables $x_e \in
[0,1]$ for $e \in E \setminus H$:

\begin{align*}
    \min \sum_{e \in E \setminus H} c(e)x_e& & \text{\bf Augment-LP} \\
    \text{subject to } \sum_{e \in \delta_{E \setminus H}(S)} x_e &\geq 1 &\forall S \in \calC  \\
    x_e \in [0, 1] & & e \in E \setminus H
\end{align*}

\begin{claim}
  Augment-LP admits an $n^{O(q)}$-time separation oracle and hence can be solved in polynomial time for each fixed $q$.
\end{claim}

Recall that $(p, 0)$-Flex-SNDP is equivalent to EC-SNDP with all
terminal pairs having requirement $p$. We can obtain a
$2$-approximate feasible solution. For $\ell=0$ to $q-1$ we solve an
augmentation problem in each stage to go from $(p,\ell)$ to $(p,\ell+1)$
flex-connectivity. Thus an $\alpha$-approximation for the augmentation
problem implies an overall $(2 + q\alpha)$-approximation for $(p,
q)$-Flex-SNDP.

\begin{remark}
There is an LP relaxation for $(p,q)$-Flex-SNDP with the property that
a feasible fractional solution to it is also feasible for Augment-LP
for each stage \cite{ChekuriJ22}. Thus, proving an integrality gap
bound for Augment-LP gives upper bounds on the integrality gap of the
LP for $(p,q)$-Flex-SNDP.
\end{remark}

\subsection{\racke Tree Embeddings}
The results in this paper use \racke's capacity-based probabilistic
tree embeddings. We borrow the notation from~\cite{ChenLLZ22}. Given $G = (V, E)$ with capacity $x: E \to \R^+$
on the edges, a capacitated tree embedding of $G$ is a tree $\calT$,
along with two mapping functions $\calM_1: V(\calT) \rightarrow V(G)$
and $\calM_2: E(\calT) \rightarrow 2^{E(G)}$ that satisfy some
conditions. $\calM_1$ maps each vertex in $\calT$ to a vertex in $G$,
and has the additional property that it gives a one-to-one mapping
between the leaves of $\calT$ and the vertices of $G$.  $\calM_2$ maps
each edge $(a,b) \in E(\calT)$ to a path in $G$ between $\calM_1(a)$
and $\calM_1(b)$.  For notational convenience we view the two mappings
as a combined mapping $\calM$. For a vertex $u \in V(G)$ we use
$\calM^{-1}(u)$ to denote the leaf in $\calT$ that is mapped to $u$ by
$\calM_1$. For an edge $e \in E(G)$ we use $\calM^{-1}(e) = \{f \in
E(\calT)\mid e \in \calM_2(f)\}$. It is sometimes convenient to
view a subset $S \subseteq V(G)$ both as vertices in $G$ and
also corresponding leaves of $\calT$.

The mapping $\calM$ induces a capacity function $y: E(\calT)
\rightarrow \mathbb{R}_+$ as follows. Consider $f=(a,b) \in
E(\calT)$. $\calT - f$ induces a partition $(A,B)$ of $V(T)$ which in
turn induces a partition/cut $(A',B')$ of $V(G)$ via the mapping
$\calM$: $A'$ is the set of vertices in $G$ that correspond to the
leaves in $A$ and similarly $B'$. We then set $y(f) = \sum_{e \in
  \delta(A')} x(e)$, in other words $y(f)$ is the capacity of cut
$(A',B')$ in $G$. The mapping also induces loads on the edges of
$G$. For each edge $e \in G$, we let $\load(e) = \sum_{f \in E(\calT):
  e \in M(f)} y(f)$.  The relative load or \emph{congestion} of $e$ is
$\rload(e) = \load(e)/x(e)$.  The congestion of $G$ with respect to a
tree embedding $(\calT, \calM)$ is defined as $\max_{e \in E(G)}
\rload(e)$. Given a probabilistic distribution $\calD$ on trees
embeddings of $(G,x)$ we let
$$\beta_{\calD} = \max_{e \in E(G)}  \E_{(\calT,\calM) \sim \calD} \rload(e)$$
denote the maximum expected congestion.
\racke showed the following fundamental result on probabilistic embeddings of a capacitated graph into trees.

\begin{theorem}[\cite{Racke08}]
\label{thm:racke}
Given a graph $G$ and $x: E(G) \to \R^+$, there exists a probability
distribution $\calD$ on tree embeddings such that $\beta_{\calD} =
O(\log |V(G)|)$.  All trees in the support of $\calD$ have height at
most $O(\log(nC))$, where $C$ is the ratio of the largest to smallest
capacity in $x$. Moreover, there is randomized polynomial-time
algorithm that can sample a tree from the distribution $\calD$.
\end{theorem}

In the rest of the paper we use $\beta$ to denote the guarantee
provided by the preceding theorem where $\beta =
O(\log n)$ for a graph on $n$ nodes.

\paragraph{Implication for flows:} The original motivation for capacitated tree embeddings
is oblivious routing of multicommodity flows. A multicommodity flow
instance in a graph $G=(V,E)$ is specified by a demand matrix $D: V
\times V \rightarrow \mathbb{R}_+$ and the goal is to simultaneously
route $D(u,v)$ amount of flow between $u$ and $v$ for each vertex pair
$(u,v)$. We say that $D$ is routable in $G$ with congestion $\alpha$
if there is a feasible multicommodity flow in $G$ that satisfies all
demands such that the total flow on each edge $e$ is at most $\alpha
\cdot x(e)$ (it is routable if $\alpha \le 1$). It can be seen that
given a tree $(\calT,\calM)$ in $\calD$, any multicommodity flow that
can be routed in $G$ with capacities $x$ can also be routed in $\calT$
with capacities given by $y$ with congestion $1$ --- this is because
cut-condition is necessary for routing and in trees it is also
sufficient. Moreover, any routable multicommodity flow in $\calT$ with
demands only between leaves can be routed in $G$ with congestion
$\max_{e \in E(G)} \rload(e)$. The mapping of the routing in $\calT$
to $G$ is simple and follows the paths given by $\calM$. The
implication of this connection is the following corollary where we use
$\maxflow_H^z(A,B)$ to denote the maxflow between two disjoint vertex
subsets $A,B$ in a capacitated graph $H$ with capacities given by $z:
E(H) \rightarrow \mathbb{R}_+$.

\begin{corollary}
\label{cor:racketreeflow}
Let $\calD$ be the distribution guaranteed in Theorem~\ref{thm:racke}.
Let $A, B \in V(G)$ be two disjoint sets. Then (i) for any tree
$(\calT,\calM)$ in $\calD$, $\maxflow_G^x(A, B) \leq
\maxflow_\calT^y(\calM^-(A), \calM^-(B))$ and (ii) $\frac 1 \beta
\E_{(\calT,\calM) \sim \calD}[\maxflow_\calT^y(\calM^-(A),
  \calM^-(B))] \leq \maxflow_G^x(A, B)$.
\end{corollary}

\subsection{Group Steiner Tree, Set Connectivity and Tree Rounding}
The group Steiner tree problem was introduced in \cite{ReichW89} and
studied in approximation by Garg, Konjevod and Ravi
\cite{GargKR98}. The input is an edge-weighted graph $G=(V,E)$, a root
vertex $r \in V$, and $k$ groups $S_1,S_2,\ldots,S_k$ where each $S_i
\subseteq V$. The goal is to find a min-weight subgraph $H$ of $G$
such there is a path in $H$ from $r$ to each group $S_i$ (that is, to
some vertex in $S_i$). The approximability of this problem has
attracted substantial attention.  Garg et al.\ \cite{GargKR98}
described a randomized algorithm to round a fractional solution to a
cut-based LP relaxation when $G$ is a tree --- it achieves a $O(\log n
\log k)$-approximation. This has been shown to be essentially tight
from both an integrality gap and a hardness point of view
\cite{HalperinKKSW07,HalperinK03}. Their algorithm also yields an
$O(\log^2 n \log k)$-approximation in general graphs via
embeddings into tree metrics \cite{Bartal98,FRT03}. Better approximation in quasi-polynomial time
are known \cite{ChekuriP05,GrandoniLL19,GhugeN22}.

Set Connectivity is a generalization of group Steiner tree problem. Here we are given pairs of sets
$(S_1,T_1),(S_2,T_2),\ldots,(S_k,T_k)$ and the goal is to find a
min-cost subgraph $H$ such that there is an $(S_i,T_i)$ path in $H$
for each $i$.
Chekuri et al.\ \cite{ChekuriEGS11} obtained a
poly-logarithmic approximation and integrality gap by generalizing the
ideas from group Steiner tree.
Chalermsook, Grandoni and Laekhanukit \cite{CGL15} studied Survivable
Set Connectivity problem, motivated by earlier work in
\cite{GuptaKR10}. Here each pair $(S_i,T_i)$ has a connectivity
requirement $r_i$ which implies that one seeks $r_i$ edge-disjoint
paths between $S_i$ and $T_i$ in the chosen subgraph $H$; \cite{CGL15}
obtained a bicriteria-approximation via \racke tree and group Steiner tree rounding.
The recent work of Chen et al \cite{ChenLLZ22} uses related but more sophisticated ideas
to obtain the first true approximation for this problem. They refer to
the problem as Group Connectivity problem and obtain an $O(r^3 \log
r \log^7 n)$-approximation where $r = \max_i r_i$ connectivity
requirement (see \cite{ChenLLZ22} for more precise bounds).

\paragraph{Oblivious tree rounding:} The rounding algorithm for Set Connectivity
in trees given in \cite{ChekuriEGS11} establishes a poly-logarithmic
integrality gap, however, the rounding is not \emph{oblivious} to the
pairs. In \cite{CGL15} a randomized oblivious algorithm based on the
group Steiner tree rounding from \cite{GargKR98} is described. This is
useful since the sets to be connected during the course of their
algorithm are implicitly generated. We encapsulate their result in the
following lemma. The tree rounding algorithm in \cite{CGL15,ChenLLZ22}
is phrased slightly differently since they combine aspects of group
Steiner rounding and the congestion mapping that comes from \racke
trees. We separate these two explicitly to make the idea more
transparent. We refer to the algorithm from the lemma below as TreeRounding.

\begin{lemma}[\cite{CGL15,ChenLLZ22}]
\label{lem:setconnectivity-tree-rounding}
Consider an instance of Set Connectivity on an $n$-node tree $T=(V,E)$
with height $h$ and let $x: E \rightarrow [0,1]$. Suppose $A, B
\subseteq V$ are disjoint sets and suppose $K \subseteq E$ such that
$x$ restricted to $K$ supports a flow of $f \le 1$ between $A$ and
$B$. There is a randomized algorithm that is oblivious to $A, B, K$
(hence depends only on $x$ and value $f$) that outputs a subset $E' \subseteq E$
such that (i) The probability that $E' \cap K$ connects $A$ to $B$ is
at least a fixed constant $\phi$ and (ii) For any edge $e \in E$, the
probability that $e \in E'$ is $\min\{1,O(\frac{1}{f} h \log^2 n) x(e)\}$.
\end{lemma}

%% file: algorithm.tex
\section{Rounding Algorithm for the Augmentation Problem}
\label{sec:algo}

\newcommand{\LG}{\textnormal{LARGE}}
\newcommand{\SM}{\textnormal{SMALL}}
\newcommand{\lpopt}{\text{OPT}_{\text{LP}}} We adapt the algorithm and
analysis in~\cite{ChenLLZ22} to Flex-SNDP.  Let $\beta$ be the
expected congestion given by Theorem~\ref{thm:racke}.  Consider an
instance of $(p,q+1)$-Flex-SNDP specified by a graph $G$ and a set of
pairs $(s_i,t_i)$, $i \in [k]$. Assume that we have a partial solution
$H$ in which each $(s_i,t_i)$ is $(p,q)$-flex-connected. We augment
$H$ to ensure that each $(s_i,t_i)$ is $(p,q+1)$-flex-connected.

We start by obtaining a solution $\{x_e\}_{e \in E \setminus H}$ for
the LP relaxation to the augmentation problem, {\bf Augment-LP},
described in Section~\ref{sec:prelim}. Let $E' = E \setminus H$.  We
define $\LG = \{e \in E': x_e \geq \frac 1 {4(p+q)\beta}\}$, and
$\textnormal{SMALL} = \{e \in E: x_e < \frac 1 {4(p+q)\beta}\}$.  The
LP has paid for each $e \in \LG$ a cost of $c(e)/(4(p+q)\beta)$, hence
adding all of them to $H$ will cost $O((p+q)\beta \cdot \lpopt)$.  If
$\LG \cup H$ is a feasible solution to the augmentation problem, then
we are done since we obtain a solution of cost $O((p+q)\log n \cdot
\lpopt)$.  Thus, the interesting case is when $\LG \cup H$ is \emph{not} a
feasible solution.  In effect, we can assume that $\LG =
\emptyset$. One can assume this situation without loss of generality
by creating sufficiently many parallel edges for each $e$ and
splitting the $x(e)$ amongst them.

Following \cite{ChenLLZ22} we employ a \racke tree based rounding.
A crucial step is to set up a capacitated graph appropriately.  We can
assume, with a negligible increase in the fractional cost, that for
each edge $e \in E \setminus H$, $x(e) = 0$ or $x(e) \ge
\frac{1}{n^3}$; this can be ensured by rounding down to $0$ the
fractional value of any edge with very small value, and compensating
for this loss by scaling up the fractional value of the other edges by
a factor of $(1+1/n)$. It is easy to check that the new solution
satisfies the cut covering constraints, and we have only increased the
cost of the fractional solution by a $(1+1/n)$-factor. In the
subsequent steps we can ignore edges with $x_e = 0$ and assume that
there are no such edges.

Consider the original graph $G=(V,E)$ where we set a capacity for each
$e \in E$ as follows. If $e \in \LG \cup H$ we set $\tilde x_e = \frac
1 {4(p+q)\beta}$. Otherwise we set $\tilde x_e = x_e$.  Since the
ratio of the largest to smallest capacity is $O(n^3)$, the height of
any \racke tree for $G$ with capacities $\tilde x$ is at most
$O(\log n)$.  Then, we repeatedly sample \racke trees. For each
tree, we sample edges by the rounding algorithm given by Chalermsook
et al in~\cite{CGL15} (see Section~\ref{sec:prelim} for details).  A
formal description of the algorithm is provided below where $t'$ and
$t$ are two parameters that control the number of trees sampled and
the number of times we run the tree rounding algorithm in each sampled
tree.  We will analyze the algorithm by setting both $t$ and $t'$ to $\Theta((p+q)\log n)$.

\begin{algorithm}[H]
\caption{Augmentation from $(p,q)$ to $(p, q+1)$-flex-connectivity}
\label{augmentation_algo}
\begin{algorithmic}
    \State $H \gets$ partial solution satisfying $(p,q)$-flex-connectivity for given instance
    \State $\{x\}_{e \in E} \gets$ fractional solution to \textbf{Augment-LP}
    \State $\textnormal{LARGE} \gets \{e \in E: x_e \geq \frac 1 {4(p+q)\beta}\}$
    \State $\textnormal{SMALL} \gets \{e \in E: x_e < \frac 1 {4(p+q)\beta}\}$
    \State $H \gets H \cup \textnormal{LARGE}$
    \If{$H$ is a feasible solution to $(p, q+1)$-Flex-SNDP}
        \Return $H$
    \Else 
        \State $\tilde x_e \gets \begin{cases}
            \frac 1 {4(p+q)\beta} & e \in H \\
            0 & x_e < \frac 1 {n^3} \\
            x_e & \text{otherwise}
            \end{cases}$
    \EndIf
    \State $\calD \gets $ \racke tree distribution for $(G, \tilde x)$
    \For{$i = 1, \dots t'$}
        \State Sample a tree $(\calT, \calM, y) \sim \calD$
        \For{$j = 1, \dots, t$}
            \State $K \gets $ output of oblivious TreeRounding algorithm on $(G, \calT)$
            \State $H \gets H \cup \calM(K)$
        \EndFor
    \EndFor \\
    \Return H
\end{algorithmic}
\end{algorithm}

%% file: analysis.tex
\section{Analysis}
\label{sec:analysis}
We will assume, following earlier discussion, that $\LG = \emptyset$
and focus on the case when the algorithm proceeds to the TreeRounding
step. Let $H$ denote the set of edges that satisfies
$(p,q)$-flex-connectivity for the given pairs. {\bf Augment-LP} is a
cut covering LP. Consider any violated cut $S$ with respect to $H$;
$S$ is violated because $S$ separates a pair $(s_i,t_i)$ and
$\delta_H(S)$ has exactly $(p+q)$ edges, of which at most $p-1$ are
safe. Let $F= \delta_H(S)$. We call $F$ a violating edge set. There
are at most $\binom {|H|} {p+q}$ violating edge sets, and since $|H|
\le n^2$, this is upper bounded by $O(n^{2(p+q)})$.
We say that a set of edges $H' \subseteq E \setminus H$ is a feasible
augmentation for violating edge set $F$ if for each pair $(s_i,t_i)$,
there is a path from $s_i$ to $t_i$ in the graph $(H\cup H')\setminus
F$. The following is a simple observation.

\begin{claim}
$H' \subseteq E \setminus H$ is a feasible solution to the
  augmentation problem iff for each violating edge set $F$, $H'$ is a
  feasible augmentation for $F$.
\end{claim}

The preceding observations allows us to focus on a fixed violating
edge set $F$, and ensuring that the algorithm outputs a set $H'$ that
is a feasible augmentation for $F$ with high probability. We observe
that the algorithm is oblivious to $F$. Thus, if we obtain a high
probability bound for a fixed $F$, since there are $O(n^{2(p+q)})$
violating edge sets, we can use the union bound to argue that $H'$ is
feasible solution for \emph{all} violating edge sets.  For the
remainder of this section, until we do the final cost analysis, we
work with a fixed violating edge set $F$.

Consider a tree $(\calT,\calM, y)$ in the \racke distribution for the
graph $G$ with capacities $\tilde x$. We let $\calM^{-1}(F)$ denote
the set of all tree edges corresponding to edges in $F$,
i.e. $\calM^{-1}(F) = \cup_{e \in F} \calM^{-1}(e)$.  We call $(\calT,
\calM, y)$ \emph{good} with respect to $F$ if $y(\calM^{-1}(F)) \leq
\frac 1 2$; equivalently, $F$ blocks a flow of at most $\frac 1 2$ in
$\calT$.

\begin{lemma}
\label{lemma:good_tree}
For a violating edge set $F$, a randomly sampled R\"{a}cke tree $(\calT, \calM, y)$ is good with respect to $F$ with probability at least $\frac 1 2$.
\end{lemma}
\begin{proof}
For each $e \in F$, $\tilde x_e =\frac 1 {4(p+q)\beta}$. Since the expected congestion of each edge is at most $\beta$, $\E[\load(e)] \leq \beta \tilde x_e \le \frac 1 {4(p+q)}$ for each $e \in F$. Note that $y(\calM^{-1}(F)) = \sum_{e \in F} \load(e)$, hence by linearity of expectation, $\E[y(\calM^{-1}(F)] = \sum_{e \in F} \E[\load(e)] \leq |F|\frac 1 {4(p+q)} = \frac 1 4$. Applying Markov's inequality to $y(\calM^{-1}(F))$ proves the lemma.
\end{proof}

Given the preceding lemma, a natural approach is to sample a good tree $\calT$ and hope that $\calT \setminus M^{-1}(F)$ still has good flow between each terminal pair. However, since we rounded down all edges in $\LG \cup H$, it is possible that $\calM^{-1}(F)$ contains an edge whose removal would disconnect a terminal pair in $\calT$, even if $\calT$ is good. See \cite{ChenLLZ22} for a more detailed discussion and example.

We note that our goal is to find a set of
edges $H' \subseteq E$ such that each $s_i$ to $t_i$ has a path in
$(H' \cup H) \setminus F$; these paths must exist in the original
graph, even if they do not exist in the tree. Therefore, instead of
looking directly at paths between $s_i$ and $t_i$ in $\calT$, we focus
on obtaining paths through components that are already connected in
$(V(G), H \setminus F)$. The rest of the argument is to show that
sufficiently many iterations of TreeRounding on any good tree $\calT$
for $F$ will yield a feasible set $H'$ for $F$.

\subsection{Shattered Components, Set Connectivity and Rounding}
Let $\Q_F$ be the set of connected components in the subgraph induced by $H \setminus F$. We use vertex subsets to denote components.
Let $\calT$ be a good tree for $F$. We say that a connected component $Q \in \Q_F$ is \emph{shattered} if it is disconnected in $\calT \setminus \calM^{-1}(F)$, else we call it \emph{intact}. For each $i \in [k]$, let $Q_{s_i} \in \Q_F$ be the component containing $s_i$, and $Q_{t_i} \in \Q_F$ be the component containing $t_i$. Note that $Q_{s_i}$ may be the same as $Q_{t_i}$ for some $i$, but if $F$ is a violating edge set then there is at least one $i$ such that $Q_{s_i} \neq Q_{t_i}$. Now, we define a Set Connectivity instance that is induced by $F$ and $\calT$. Consider two disjoint vertex subsets $A,B \subset V$.
We say that $(A,B)$ partitions the set of shattered components if each shattered component $Q$ is fully contained in $A$ or fully contained in $B$. 
Formally let 
$$Z_F = \{(A \cup Q_{s_i}, B \cup Q_{t_i}): (A, B)
\text{ partitions the shattered components}, i \in [k]\}.$$ In other
words, $Z_F$ is set of all partitions of shattered components that
separate some pair $(s_i,t_i)$.  Since the leaves of $\calT$ are in
one to one correspondence with $V$ we can view $Z_F$ as inducing a Set
Connectivity instance in $\calT$; technically we need to consider the
pairs $\{(\calM^{-1}(A),\calM^{-1}(B)) \mid (A,B) \in Z_F\}$; however,
for simplicity we conflate the leaves of $\calT$ with
$V$.  We claim that it suffices to find a feasible solution that
connects the pairs defined by $Z_F$ in the tree $\calT$.

\begin{claim}
\label{claim:shattered_suffices}
Let $E' \subseteq \calT \setminus \calM^{-1}(F)$. Suppose there exists a path in $E' \subseteq \calT \setminus \calM^{-1}(F)$ connecting $A$ to $B$ for all $(A, B) \in Z_F$. Then, there is an $s_i$-$t_i$ path for each $i \in [k]$ in $(\calM(E') \cup H) \setminus F$. 
\end{claim}
\begin{proof}
Let $E' \subseteq \calT \setminus \calM^{-1}(F)$ such that there is a path from $A$ to $B$ in $E'$ for each $(A, B) \in Z_F$. Assume for the sake of contradiction that $\exists i \in [k]$ such that $(s_i, t_i)$ are disconnected in $(\calM(E') \cup H) \setminus F$. Then, there must be some cut $S$ such that $\delta_{(\calM(E') \cup H) \setminus F}(S) = \emptyset$ and $|S \cap \{s_i, t_i\}| = 1$.

We observe that no component $Q \in \Q_F$ can cross $S$ since each $Q$
is connected in $H\setminus F$. Assume without loss of generality that
$s_i \in S$. Then, let $A = Q_{s_i} \cup \{Q \in \Q_F: Q \text{ is
  shattered }, Q \subseteq S\}$, and $B = Q_{t_i} \cup \{Q \in \Q_F: Q
\text{ is shattered }, Q \subseteq \overline S\}$. Clearly, $A
\subseteq S$, $B \subseteq \overline S$. Furthermore, $(A, B) \in
Z_F$. By
assumption, there is a path $P$ in $E'$ between $A$ and $B$. Since $E'
\cap \calM^{-1}(F) = \emptyset$, $\calM(E')$ cannot contain any edges
in $F$. Therefore, $\calM(P)$ contains a path that crosses $S$ which
implies that $|\delta_{\calM(E')}(S)| = |\delta_{\calM(E') \setminus
  F}(S)| \geq 1$, contradicting the assumption on $S$.
\end{proof}

We now argue that $(\calT, \calM, y)$ routes sufficient flow for each pair in $Z_F$ without using the edges in $\calM^{-1}(F)$; in other words $y$ is fractional solution (modulo a scaling factor) to the Set Connectivity instance $Z_F$ in the graph/forest $\calT \setminus \calM^{-1}(F)$. We can then appeal to TreeRounding lemma to argue that it will connect the pairs in $Z_F$ without using any edges in $F$.

\begin{lemma}
\label{lem:flowforeachpair}
Let $(A, B) \in Z_F$. Let
$S \subset V_{\calT}$ such that $A \subseteq S$ and $B \subseteq V_{\calT} \setminus S$. Then $y(\delta_{\calT \setminus \calM^{-1}(F)}(S)) \geq \frac 1 {4(p+q)\beta}$.
\end{lemma}
\begin{proof}
Let $S$ be a vertex set of $\calT$ that separates $A$ from $B$. First,
suppose there exists a component $Q \in \Q_F$ such that $Q$ crosses
$S$, i.e. $S \cap Q \neq \emptyset$ and $\overline S \cap Q \neq
\emptyset$. Since $(A, B)$ partitions the set of shattered components,
$Q$ must be intact in $\calT$. Let $u$ be a leaf in $Q \cap S$ and $v$
be a leaf in $Q \cap \overline S$.  Since $Q$ is intact in $T$ the
unique path connecting $u$ to $v$ in $\calT$ crosses $S$ and let $e$
be an edge on this path that crosses $S$. It suffices to show that
$y(e) \ge \frac 1 {4(p+q)\beta}$. This follow from properties of the
\racke tree. Since $u$ and $v$ are connected in $G'$ with a path using
only edges in $\LG \cup H$ each of which has a capacity of $\frac 1
{4(p+q)\beta}$, $u$-$v$ maxflow in $G'$ is at least $\frac 1
{4(p+q)\beta}$. From Corollary~\ref{cor:racketreeflow},
for any tree $\calT$, the $u$-$v$
maxflow in $\calT$ with capacities $y$ must be at least $\frac 1
{4(p+q)\beta}$. This in particular implies that $y(e) \ge \frac 1
{4(p+q)\beta}$ for every edge $e$ on the unique path from $u$ to $v$
in $\calT$.

 We can now restrict attention to the case that no connected component
 of $\Q_F$ crosses $S$. Consider $S'$ be the set of leaves in $S$ and
 consider the cut $(S',V\setminus S')$ in $G$. It follows that
 $(S',V-S')$ partitions the connected components in $\Q_F$ and
 $\delta_{H-F}(S') = \emptyset$. Since $(A,B) \in Z_F$ there is a pair
 $(s_i,t_i)$ such that $Q_{s_i} \in S'$ and $Q_{t_i} \in V\setminus
 S'$. Thus $(S',V\setminus S')$ is a violated cut with $F$ as its
 witness. Since $x$ is a feasible solution to {\bf Augment-LP} it follows
 that $x(\delta_{E \setminus H}(S')) \ge 1$. Recall that we assumed
 that $\LG = \emptyset$, and hence all
 edges in $\delta_{E \setminus H)}(S')$ are in $\SM$.
 Therefore, $x(\delta_{E \setminus H}(S')) = \tilde x(\delta_{E \setminus H}(S')) \ge 1$.

 The \racke tree property guarantees that $y(\delta_{\calT}(S)) \ge \tilde x(\delta_{G'}(S')) \ge 1$ (via Corollary~\ref{cor:racketreeflow}). 
 We note that 
 $$y(\delta_{\calT \setminus \calM^{-1}(F)}(S)) \ge y(\delta_{\calT}(S)) - y(\calM^-(F)) \ge 1 - 1/2 \ge 1/2.$$
 where we used the fact that $y(\calM^-(F)) \le 1/2$ since $\calT$ is good for $F$.
 Thus in both cases we verify the desired bound.
\end{proof}

\paragraph{Bounding $Z_F$:} A second crucial property is a bound on $|Z_F|$,
the number of pairs in the Set Connectivity instance induced by $F$ and a good tree $\calT$ for $F$. 

\begin{lemma}
\label{lem:boundnumpairs}
For a good tree $\calT$, $|Z_F| \leq 2^{2(p+q)\beta} k$.
\end{lemma}
\begin{proof}
Let $\ell$ be the number of shattered components and let them be
$Q_1,\ldots,Q_\ell$. For each $Q_i$ pick a pair of vertices $u_i,v_i$
that are in separate components of $\calT - \calM^{-1}(F)$. Let $A =
\{u_1,\ldots,u_\ell\}$ and $B = \{v_1,v_2,\ldots,v_\ell\}$. Since the
paths connecting $u_i,v_i$ are in different connected components of
$H\setminus F$, it follows that the $(A,B)$-maxflow in $H \setminus F$
is at least $\ell$. In the graph $G'$ obtained by scaling down the
capacity of edges of $H$, the maxflow is at least $\frac \ell
{4(p+q)\beta}$ which implies that it is at least this quantity in
$\calT$. Since $\calT$ is good, the total decrease of flow can be at
most $y(\calM^{-1}(F)) \le \frac 1 2$. By construction there is no
flow between $A$ and $B$ in $\calT - \calM^{-1}(F)$ which implies that
$\frac \ell {4(p+q)\beta} \le 1/2 \Rightarrow \ell \leq 2(p+q)\beta$.
Each pair in $Z_F$ corresponds to a subset of shattered components and
a demand pair $(s_i,t_i)$, and hence $|Z_F|\leq 2^\ell k \le
2^{2(p+q)\beta} k$.
\end{proof}

\subsection{Correctness and Cost}
Now we analyze the correctness and cost of the algorithms output.

\begin{lemma}
\label{claim:successprobforgoodtree}
Suppose $\calT$ is good for a violating edge set $F$. Then after $t$
rounds of TreeRounding with flow parameter $\frac 1 {4(p+q)\beta}$,
the probability that $H'$ is \emph{not} a feasible augmentation for
$F$ is at most $(1-\phi)^t |Z_F| \le 1/4$.
\end{lemma}
\begin{proof}
  Suppose $\calT$ is good for $F$. Let $(A,B) \in Z_F$. From
  Lemma~\ref{lem:flowforeachpair} the flow for $(A,B)$ in $\calT -
  \calM^{-1}(F)$ is at least $\frac 1 {4(p+q)\beta}$. From
  Lemma~\ref{lem:setconnectivity-tree-rounding}, with probability at
  least $\phi$, the pair $(A,B)$ is connected via a path in $\calT -
  \calM^{-1}(F)$. If all pairs are connected, then via
  Claim~\ref{claim:shattered_suffices}, $H'$ is a feasible
  augmentation for $F$. Thus, $H'$ is not a feasible augmentation if
  for some $(A,B) \in Z_F$ the TreeRounding does not succeed after $t$
  rounds. The probability of this, via the union bound over the pairs
  in $Z_F$, is at most $(1-\phi)^t |Z_F|$. From
  Lemma~\ref{lem:boundnumpairs}, $|Z_F| \le 2^{2(p+q)\beta}k$. Consider $t
  = \frac 1 \phi \log(4k \cdot 2^{2\beta(p+q)}) = O((p+q)\log n)$, since $\beta = O(\log n)$. Then, $(1-\phi)^t |Z_F| = 2^{2(p+q)\beta} k (1 -
  \phi)^t \leq 2^{2(p+q)\beta} k e^{-\phi t} \leq \frac 1 4$.
\end{proof}

\begin{lemma}
\label{lem:correctness}
The algorithm outputs a solution $H'$ such that $H \cup H'$ is a
feasible augmentation to the given instance with probability at least
$\frac 1 2$.
\end{lemma}
\begin{proof}
For a fixed $F$ the probability that a sampled tree is good is at
least $1/2$.  By Claim~\ref{claim:successprobforgoodtree}, conditioned
on the sampled tree being good for $F$, $t$ iterations of TreeRounding
fail to augment $F$ with probability at most $1/4$. Thus the probability that
all $t'$ iterations of sampling trees fail is $(1-3/8)^{t'}$. There
are at most $n^{2(p+q)}$ violating edge sets $F$. Consider $t' = \frac 8 3
\log (2n^{2(p+q)}) = O((p+q)\log n)$. By applying the union bound over
all violating edge sets $F$, the probability of the algorithm failing
is at most $n^{2(p+q)}(1 - 3/8)^{t'} \leq n^{2(p+q)}e^{-3t'/8} \leq
\frac 1 2$. Therefore, the output of the algorithm is a feasible
augmentation for all violating edge sets with probability at least
$\frac 1 2$.
\end{proof}

Now we analyze the expected cost of the edges output by the algorithm for augmentation
with respect to $\lpopt$, the cost of the fractional solution.

\begin{lemma}
\label{lem:costanalysis}
The total expected cost of the algorithm is $O((p+q)^3\log^7 n) \cdot \lpopt$.
\end{lemma}
\begin{proof}
Fix an edge $e \in \SM$ with fractional value $x_e$. Consider one
outer iteration of the algorithm in which it picks a random tree
$\calT$ from the \racke tree distribution and then runs $t$ iterations
of TreeRounding with flow parameter $\alpha = \frac 1
{4(p+q)\beta}$. Via Lemma~\ref{lem:setconnectivity-tree-rounding}, the
probability of an edge $f \in \calT$ being chosen is at most
$O(\frac{1}{\alpha} h \log^2n) y(f)$. Thus the expected cost for $e$
for one round of TreeRounding is $O(\frac{1}{\alpha} h \log^2n)
\sum_{f \in \calM^{-1}(e)} y(f) = O(\frac{1}{\alpha} h \log^2n)
\load(e)$. By the \racke distribution property, $\E_{\calT} [\load(e)]
\le \beta x_e$. By linearity of expectation, since there are a total
of $t \cdot t'$ iterations of TreeRounding, the total expected cost is
at most ($t \cdot t')\cdot O(\frac{1}{\alpha} h \log^2 n \beta) \sum_{e \in
  E} c(e) x_e$. By the analysis in Section~\ref{sec:algo}, $h = O(\log
n)$, and $\beta = O(\log n)$. Substituting in the values of $t$ and $t'$ stated
in Lemmas \ref{claim:successprobforgoodtree} and \ref{lem:correctness}, the total expected cost is at most $O((p+q)^3 \log^7 n) \cdot \lpopt$.
\end{proof}

Combining the correctness and cost analysis we obtain the following.
\begin{theorem}
  There is a randomized $O((p+q)^3\log^7 n)$ approximation for the augmentation problem
  via {\bf Augment-LP}.
\end{theorem}

Starting with a solution for $(p,0)$-flex-connectivity, and using $q$ augmentation iterations,
we obtain an $O(q(p+q)^3\log^7 n)$-approximate solution for the given instance of
$(p,q)$-Flex-SNDP, proving
Theorem~\ref{thm:main}.

%% file: paper.bbl
\begin{thebibliography}{10}

\bibitem{Adjiashvili13}
David Adjiashvili.
\newblock Fault-tolerant shortest paths - beyond the uniform failure model,
  2013.
\newblock URL: \url{https://arxiv.org/abs/1301.6299}, \href
  {https://doi.org/10.48550/ARXIV.1301.6299}
  {\path{doi:10.48550/ARXIV.1301.6299}}.

\bibitem{AdjiashviliHM20}
David Adjiashvili, Felix Hommelsheim, and Moritz M{\"u}hlenthaler.
\newblock Flexible graph connectivity.
\newblock In {\em International Conference on Integer Programming and
  Combinatorial Optimization}, pages 13--26. Springer, 2020.

\bibitem{AdjiashviliHM22}
David Adjiashvili, Felix Hommelsheim, and Moritz M{\"u}hlenthaler.
\newblock Flexible graph connectivity.
\newblock {\em Mathematical Programming}, 192(1):409--441, 2022.

\bibitem{AdjiashviliHMS20}
David Adjiashvili, Felix Hommelsheim, Moritz M{\"{u}}hlenthaler, and Oliver
  Schaudt.
\newblock Fault-tolerant edge-disjoint $s$-$t$ paths - beyond uniform faults.
\newblock In {\em 18th Scandinavian Symposium and Workshops on Algorithm
  Theory, {SWAT}}, volume 227 of {\em LIPIcs}, pages 5:1--5:19, 2022.
\newblock See \url{https://arxiv.org/abs/2009.05382} for a full version from
  2020.
\newblock \href {https://doi.org/10.4230/LIPIcs.SWAT.2022.5}
  {\path{doi:10.4230/LIPIcs.SWAT.2022.5}}.

\bibitem{BansalCGI22}
Ishan Bansal, Joseph Cheriyan, Logan Grout, and Sharat Ibrahimpur.
\newblock Approximating $(p,2)$ flexible graph connectivity via the primal-dual
  method, 2022.
\newblock URL: \url{https://arxiv.org/abs/2209.11209}, \href
  {https://doi.org/10.48550/ARXIV.2209.11209}
  {\path{doi:10.48550/ARXIV.2209.11209}}.

\bibitem{Bartal98}
Yair Bartal.
\newblock On approximating arbitrary metrices by tree metrics.
\newblock In {\em Proceedings of the thirtieth annual ACM symposium on Theory
  of computing}, pages 161--168, 1998.

\bibitem{BoydCHI22}
Sylvia Boyd, Joseph Cheriyan, Arash Haddadan, and Sharat Ibrahimpur.
\newblock Approximation algorithms for flexible graph connectivity, 2022.
\newblock A preliminary version of the paper appeared in Proc.\ of FSTTCS 2021.
\newblock URL: \url{https://arxiv.org/abs/2202.13298}, \href
  {https://doi.org/10.48550/ARXIV.2202.13298}
  {\path{doi:10.48550/ARXIV.2202.13298}}.

\bibitem{CGL15}
Parinya Chalermsook, Fabrizio Grandoni, and Bundit Laekhanukit.
\newblock {\em On Survivable Set Connectivity}, pages 25--36.
\newblock URL: \url{https://epubs.siam.org/doi/abs/10.1137/1.9781611973730.3},
  \href
  {http://arxiv.org/abs/https://epubs.siam.org/doi/pdf/10.1137/1.9781611973730.3}
  {\path{arXiv:https://epubs.siam.org/doi/pdf/10.1137/1.9781611973730.3}},
  \href {https://doi.org/10.1137/1.9781611973730.3}
  {\path{doi:10.1137/1.9781611973730.3}}.

\bibitem{ChekuriEGS11}
Chandra Chekuri, Guy Even, Anupam Gupta, and Danny Segev.
\newblock Set connectivity problems in undirected graphs and the directed
  steiner network problem.
\newblock {\em ACM Transactions on Algorithms (TALG)}, 7(2):1--17, 2011.

\bibitem{ChekuriJ22}
Chandra Chekuri and Rhea Jain.
\newblock Augmentation based approximation algorithms for flexible network
  design, 2022.
\newblock URL: \url{https://arxiv.org/abs/2209.12273}, \href
  {https://doi.org/10.48550/ARXIV.2209.12273}
  {\path{doi:10.48550/ARXIV.2209.12273}}.

\bibitem{ChekuriP05}
Chandra Chekuri and Martin Pal.
\newblock A recursive greedy algorithm for walks in directed graphs.
\newblock In {\em 46th annual IEEE symposium on foundations of computer science
  (FOCS'05)}, pages 245--253. IEEE, 2005.

\bibitem{ChenLLZ22}
Qingyun Chen, Bundit Laekhanukit, Chao Liao, and Yuhao Zhang.
\newblock Survivable network design revisited: Group-connectivity.
\newblock 2022.
\newblock Full version of paper in Proceedings of IEEE FOCS 2022.
\newblock URL: \url{https://arxiv.org/abs/2204.13648}, \href
  {https://doi.org/10.48550/ARXIV.2204.13648}
  {\path{doi:10.48550/ARXIV.2204.13648}}.

\bibitem{FRT03}
Jittat Fakcharoenphol, Satish Rao, and Kunal Talwar.
\newblock A tight bound on approximating arbitrary metrics by tree metrics.
\newblock In {\em Proceedings of the thirty-fifth annual ACM symposium on
  Theory of computing}, pages 448--455, 2003.

\bibitem{GargKR98}
Naveen Garg, Goran Konjevod, and Ramamoorthi Ravi.
\newblock A polylogarithmic approximation algorithm for the group steiner tree
  problem.
\newblock {\em Journal of Algorithms}, 37(1):66--84, 2000.
\newblock Preliminary version in Proc.\ of ACM-SIAM SODA 1998.

\bibitem{GhugeN22}
Rohan Ghuge and Viswanath Nagarajan.
\newblock Quasi-polynomial algorithms for submodular tree orienteering and
  directed network design problems.
\newblock {\em Mathematics of Operations Research}, 47(2):1612--1630, 2022.

\bibitem{GoemansGPSTW94}
M.~X. Goemans, A.~V. Goldberg, S.~Plotkin, D.~B. Shmoys, E.~Tardos, and D.~P.
  Williamson.
\newblock Improved approximation algorithms for network design problems.
\newblock In {\em Proceedings of the fifth annual ACM-SIAM symposium on
  Discrete algorithms}, pages 223--232, 1994.

\bibitem{GrandoniLL19}
Fabrizio Grandoni, Bundit Laekhanukit, and Shi Li.
\newblock O (log2 k/log log k)-approximation algorithm for directed steiner
  tree: a tight quasi-polynomial-time algorithm.
\newblock In {\em Proceedings of the 51st Annual ACM SIGACT Symposium on Theory
  of Computing}, pages 253--264, 2019.

\bibitem{GuptaKR10}
Anupam Gupta, Ravishankar Krishnaswamy, and Ramamoorthi Ravi.
\newblock Tree embeddings for two-edge-connected network design.
\newblock In {\em Proceedings of the Twenty-First Annual ACM-SIAM Symposium on
  Discrete Algorithms}, pages 1521--1538. SIAM, 2010.

\bibitem{HalperinKKSW07}
Eran Halperin, Guy Kortsarz, Robert Krauthgamer, Aravind Srinivasan, and Nan
  Wang.
\newblock Integrality ratio for group steiner trees and directed steiner trees.
\newblock {\em SIAM Journal on Computing}, 36(5):1494--1511, 2007.

\bibitem{HalperinK03}
Eran Halperin and Robert Krauthgamer.
\newblock Polylogarithmic inapproximability.
\newblock In {\em Proceedings of the thirty-fifth annual ACM symposium on
  Theory of computing}, pages 585--594, 2003.

\bibitem{Jain01}
K.~Jain.
\newblock A factor 2 approximation algorithm for the generalized {Steiner}
  network problem.
\newblock {\em Combinatorica}, 21(1):39--60, 2001.

\bibitem{Racke08}
Harald R\"{a}cke.
\newblock Optimal hierarchical decompositions for congestion minimization in
  networks.
\newblock STOC '08, page 255–264, New York, NY, USA, 2008. Association for
  Computing Machinery.
\newblock \href {https://doi.org/10.1145/1374376.1374415}
  {\path{doi:10.1145/1374376.1374415}}.

\bibitem{ReichW89}
Gabriele Reich and Peter Widmayer.
\newblock Beyond steiner's problem: A vlsi oriented generalization.
\newblock In {\em International Workshop on Graph-theoretic Concepts in
  Computer Science}, pages 196--210. Springer, 1989.

\end{thebibliography}
